\documentclass[11pt]{article}

\usepackage[margin=1in]{geometry}
\usepackage{amsmath,amsthm,amssymb}
\usepackage{microtype}
\usepackage{natbib}
\usepackage{xcolor}
\usepackage[hidelinks]{hyperref}

\newtheorem{proposition}{Proposition}

\newcommand{\1}{\mathbf{1}}
\newcommand{\Pp}{\mathbb{P}}

\title{Using Pre-Trends for Inference in Difference-in-Differences}

\author{Cl\'ement de Chaisemartin%
\thanks{Sciences Po, Department of Economics.
E-mail: \href{mailto:clement.dechaisemartin@sciencespo.fr}
{clement.dechaisemartin@sciencespo.fr}.
Funded by the European Union (ERC, REALLYCREDIBLE, Grant Agreement
No.\ 101043899). Views and opinions expressed are those of the author only
and do not necessarily reflect those of the European Union or the European
Research Council Executive Agency. Neither the European Union nor the
granting authority can be held responsible for them.}}

\date{}

\begin{document}

\maketitle

\section{Introduction}

Difference-in-differences (DID) are sometimes estimated with many pre-treatment
periods. In such settings, the observed pre-treatment
outcome evolutions provide direct information about the magnitude of shocks that
could also occur after treatment. This paper proposes a simple inference
procedure that uses those pre-trends as the reference distribution for the
post-treatment DID.

The idea is to test a candidate treatment effect by subtracting it from the
post-treatment outcome of the treated group and comparing the resulting
absolute DID with the empirical distribution of the absolute pre-treatment
DIDs. Under the null, the adjusted post-treatment DID is an untreated
differential trend and should therefore be distributed like its pre-treatment
counterparts. If the empirical distribution of those differential trends is
uniformly consistent for a continuous limiting distribution, the resulting
rank-based test is asymptotically exact. Confidence sets follow by test
inversion.

The procedure is closely related to conformal inference, but its DID-specific
predictor leads to a distinct identifying restriction. Existing conformal DID
procedures are based on outcome levels and imply parallel trends. Here,
the prediction error is instead the change in the treated group's untreated
outcome minus the contemporaneous change in the control group's untreated
outcome. The method therefore does not require parallel trends. It requires only
that the distribution of the absolute untreated differential trend be stable
over time. This permits the mean differential trend to be nonzero, to vary over
time, and even to change sign, while imposing stability on its magnitude. The
procedure is nonparametric, easy to implement, and especially natural in
applications with a long pre-treatment time series.

\section{Preliminary result}

Let $U_1,\ldots,U_T$ satisfy
\begin{align}\label{eq:setup}
U_t=U_t(0),\qquad t=1,\ldots,T-1,\nonumber\\
U_T=U_T(0)+\tau,
\end{align}
where $\tau\in\mathbb{R}$ is an unknown, non-stochastic treatment-effect parameter we seek to
estimate. Fix $\alpha\in(0,1)$, and assume that, for every $t$ and every $u\geq0$,
\begin{align}\label{eq:stationarity}
\Pp\!\left(|U_t(0)|\leq u\right)=F(u),
\end{align}
where $F$ is a continuous cumulative distribution function.

\medskip
For any candidate value $x\in\mathbb{R}$ for $\tau$, define the empirical
distribution function
\[
\widehat F(u;x)
=
\frac{1}{T}
\sum_{t=1}^T
\1\!\left\{
\left|U_t-\1\{t=T\}x\right|\le u
\right\}.
\]
Assume that, under the null hypothesis
\[
H_0:\tau=x,
\]
\begin{align}\label{eq:unif_consistency}
\sup_{u\ge0}
\left|
\widehat F(u;x)-F(u)
\right|
\xrightarrow{P}
0.
\end{align}
Consider the test that rejects $H_0$ whenever
\[
\widehat F(|U_T-x|;x)\ge 1-\alpha.
\]

\begin{proposition}\label{prop:conform}
Suppose \eqref{eq:setup}-\eqref{eq:unif_consistency} hold. Under $H_0:\tau=x$,
\[
\Pp\!\left(
\widehat F(|U_T-x|;x)\ge1-\alpha
\right)
\longrightarrow
\alpha.
\]
Thus, the test is asymptotically exact.
\end{proposition}

\begin{proof}
Under $H_0$,
\[
|U_T-x|=|U_T(0)|.
\]
Let
\[
V=|U_T(0)|.
\]
Then $V$ has cumulative distribution function $F$. Uniform consistency
implies
\[
\left|
\widehat F(|U_T-x|;x)-F(|U_T-x|)
\right|
\le
\sup_{u\ge0}
\left|
\widehat F(u;x)-F(u)
\right|
=o_P(1).
\]
Hence,
\[
\widehat F(|U_T-x|;x)
=
F(V)+o_P(1).
\]

Since $F$ is continuous, the probability integral transform implies
\[
F(V)\sim\operatorname{Unif}(0,1).
\]
Therefore,
\[
\widehat F(|U_T-x|;x)
\xrightarrow{d}
\operatorname{Unif}(0,1),
\]
and
\[
\Pp\!\left(
\widehat F(|U_T-x|;x)\ge1-\alpha
\right)
\longrightarrow
\alpha.
\]
\end{proof}
The test in Proposition \ref{prop:conform} follows the same basic principle as conformal inference
\citep{vovkgammermanshafer2005,chernozhukov2018}: under the null, the treated observation, after adjustment
by the hypothesized treatment effect, should be indistinguishable from the
untreated observations. Accordingly, inference is based on the empirical rank
of the treated observation among the adjusted sample. Unlike existing
conformal methods, however, our procedure requires only uniform consistency of the empirical distribution function. A confidence set for $\tau$ can be obtained by inverting the test. In
particular, an asymptotic $1-\alpha$ confidence set is

\[
\mathcal C_{1-\alpha}
=
\left\{
x\in\mathbb{R}:
\widehat F(|U_T-x|;x)<1-\alpha
\right\}.
\]
By Proposition~\ref{prop:conform},
$\Pp(\tau\in\mathcal C_{1-\alpha})\to1-\alpha$.

\section{Application to difference-in-differences}

\subsection{Set-up}

There are two groups, indexed by $g\in\{1,2\}$, and dates
$t\in\{0,\ldots,T\}$. Group~1 is never treated. Group~2 is untreated at dates
$0,\ldots,T-1$ and treated at date $T$. Let $Y_{g,t}$ denote group $g$'s
observed outcome at date $t$, and let $Y_{g,t}(0)$ denote its untreated
potential outcome. The treatment effect of group~2 at date $T$ is
\[
\tau=Y_{2,T}(1)-Y_{2,T}(0),
\]
which is assumed to be non-stochastic.

\medskip
For every $t\in\{1,\ldots,T\}$, define the untreated differential trend
\[
U_t(0)
=
Y_{2,t}(0)-Y_{2,t-1}(0)
-
\bigl(Y_{1,t}(0)-Y_{1,t-1}(0)\bigr).
\]
The observed differential trend is
\[
U_t
=
Y_{2,t}-Y_{2,t-1}
-
\bigl(Y_{1,t}-Y_{1,t-1}\bigr).
\]

\subsection{Sufficient conditions for \eqref{eq:setup}-\eqref{eq:unif_consistency} to hold}

Because neither group is treated before $T$,
\[
U_t=U_t(0),\qquad t<T,
\]
whereas
\[
U_T=U_T(0)+\tau.
\]
Thus, \eqref{eq:setup} holds.

\medskip

For \eqref{eq:stationarity} to hold, it is sufficient that $|U_t(0)|$ have
the same continuous distribution at every date $t$. Continuity may be
plausible when groups are aggregate entities (e.g., states or regions) and
$Y_{g,t}(0)$ averages outcomes across many units. Importantly,
\eqref{eq:stationarity} does not imply the usual parallel-trends condition.
Provided $U_t(0)$ is integrable, parallel trends requires
$\mathbb{E}[U_t(0)]=0$ at every date $t$. In contrast,
\eqref{eq:stationarity} restricts only the distribution of $|U_t(0)|$ and
places no restriction on the time evolution of the sign of $U_t(0)$.
Consequently, $\mathbb{E}[U_t(0)]$ may be nonzero, may vary over time, and may
change sign. The restriction is nevertheless substantive: it requires the
entire distribution of the magnitude $|U_t(0)|$, including its tail behavior
and all its quantiles, to remain stable over time.

\medskip

Finally, a sufficient condition for \eqref{eq:unif_consistency} is that the
joint increment process
\[
\left(\Delta Y_{1,t}(0),\Delta Y_{2,t}(0)\right),
\qquad
\Delta Y_{g,t}(0)=Y_{g,t}(0)-Y_{g,t-1}(0),
\]
be strictly stationary and ergodic. Then
\[
U_t(0)=\Delta Y_{2,t}(0)-\Delta Y_{1,t}(0)
\]
is also stationary and ergodic. The ergodic theorem, applied to
$\1\{|U_t(0)|\leq u\}$, gives pointwise almost-sure convergence of the
empirical distribution function to $F$ at every $u$. Because $F$ is
continuous, pointwise convergence of distribution functions implies uniform
convergence, so
\[
\sup_{u\geq0}
\left|
\frac{1}{T}\sum_{t=1}^T\1\{|U_t(0)|\leq u\}
-F(u)
\right|
\xrightarrow{P}0.
\]
Under $H_0$, the empirical distribution function in the preceding display is
exactly $\widehat F(u;x)$, establishing \eqref{eq:unif_consistency}.

A more primitive sufficient condition is that the joint increment process be
strictly stationary and strongly mixing, with mixing coefficients satisfying
$\alpha(k)\to0$. For example, this follows if the two increment processes
are independent of each other as stochastic processes and each is strictly
stationary and strongly mixing with mixing coefficients converging to zero.

\subsection{The test in the difference-in-differences example}

When testing the null of no treatment effect,
the null is rejected whenever
the $T-1$-to-$T$
difference-in-differences is larger in absolute value than the
$(1-\alpha)$-quantile of the absolute pre-treatment difference-in-differences. Intuitively, the test treats the
post-treatment differential trend as an end-of-sample observation and asks
whether it is unusually large relative to the historical distribution of
differential trends. More generally, to test $H_0:\tau=x$, the outcome of group~2 at $T$ is
adjusted by $x$, and one rejects whenever
\[
\widehat F\!\left(
\left|
Y_{2,T}-Y_{2,T-1}
-
(Y_{1,T}-Y_{1,T-1})
-x
\right|;x
\right)
\geq1-\alpha.
\]

\section{Relationship to prior literature}

\subsection{Comparison with \citet{chernozhukov2021}}

The method proposed in this paper can be viewed as a simple and novel
difference-in-differences-specific choice of counterfactual predictor within
the general conformal inference framework of
\citet{chernozhukov2021}. Their framework starts from a counterfactual
prediction rule for the untreated potential outcome of the treated unit and
conducts inference by comparing the post-treatment prediction error with the
empirical distribution of analogous pre-treatment prediction errors. Our
procedure follows this logic using a particularly simple predictor tailored to
the DID setting.

\medskip
Consider testing the sharp null that the treatment effect at date $T$ equals
$x$. Under this null,
\[
Y_{2,T}(0)=Y_{2,T}-x.
\]

For each date $t$, we predict group $2$'s untreated potential outcome by
\[
\widehat Y_{2,t}(0)
=
Y_{2,t-1}
+
\left(Y_{1,t}-Y_{1,t-1}\right),
\]
namely, its previous outcome plus the contemporaneous outcome change of the
control group. Under $H_0$, the observed outcome adjusted by the hypothesized
effect is $Y_{2,t}-\1\{t=T\}x$. The associated conformity score is therefore
\begin{align*}
\left|
Y_{2,t}-\1\{t=T\}x
-
\widehat Y_{2,t}(0)
\right|
&=
\left|
Y_{2,t}-Y_{2,t-1}
-
\left(Y_{1,t}-Y_{1,t-1}\right)
-\1\{t=T\}x
\right|\\
&=
\left|U_t-\1\{t=T\}x\right|,
\end{align*}
which is precisely the statistic used in our test. Because the prediction rule
contains no unknown parameter, the estimator-stability condition in
\citet{chernozhukov2021} is automatically satisfied. The relevant substantive
requirement is therefore the distributional stability of the untreated
prediction errors, together with the weak-dependence conditions needed for the
empirical distribution of the conformity scores to be consistently estimated.

\medskip
Relative to the DID predictor considered by
\citet{chernozhukov2021}, our predictor has a key advantage: it does not imply
parallel trends. Their DID specification predicts the untreated potential
outcome of the treated group using
\[
\widehat Y_{2,t}(0)
=
\widehat\mu+Y_{1,t},
\]
where $\widehat\mu$ estimates a time-invariant difference between the two
groups. Thus, asymptotically, the prediction error is simply the untreated
outcome gap after removing a constant. 
If this prediction error has a stable distribution over time and is
integrable, then
\[
\mathbb{E}\!\left[Y_{2,t}(0)-Y_{1,t}(0)\right]
\]
must be constant over time. First-differencing this equality yields
\[
\mathbb{E}\!\left[
Y_{2,t}(0)-Y_{2,t-1}(0)
-
\left(
Y_{1,t}(0)-Y_{1,t-1}(0)
\right)
\right]
=0,
\]
which is precisely the parallel-trends assumption.

\medskip
Our predictor instead yields the prediction error
\[
Y_{2,t}(0)-\widehat Y_{2,t}(0)
=
Y_{2,t}(0)-Y_{2,t-1}(0)
-
\left(
Y_{1,t}(0)-Y_{1,t-1}(0)
\right),
\]
namely the untreated differential trend between the two groups. Stability of
the distribution of this prediction error only requires, in particular, that
\[
\mathbb{E}\!\left[
Y_{2,t}(0)-Y_{2,t-1}(0)
-
\left(
Y_{1,t}(0)-Y_{1,t-1}(0)
\right)
\right]
\]
be constant over time. This constant may be nonzero, so group $2$ may
systematically exhibit faster or slower untreated growth than group $1$.
Thus, our procedure allows stable violations of parallel trends. As discussed earlier, even stability of the full distribution of this prediction
error is stronger than necessary for the test considered in this paper.
Because the test uses only the absolute conformity scores, it is sufficient
that
\[
\left|
Y_{2,t}(0)-\widehat Y_{2,t}(0)
\right|
=
\left|
Y_{2,t}(0)-Y_{2,t-1}(0)
-
\left(
Y_{1,t}(0)-Y_{1,t-1}(0)
\right)
\right|
\]
have a stable distribution over time.

\medskip
Overall, the distinction is between two counterfactual predictors within the same framework.
Chernozhukov, W\"uthrich, and Zhu's DID predictor is based on outcome levels
and implies parallel trends, whereas ours is based on outcome changes and
requires only distributional stability of the absolute prediction errors, or,
equivalently, of the absolute untreated differential trends.

\subsection{Related difference-in-differences literature}

The proposal in this paper is related to, but distinct from, several strands of the
difference-in-differences literature.

\medskip
\citet{conleytaber2011} and \citet{fermanpinto2019} study inference with few
treated groups. Their procedures use the cross-sectional distribution of
estimated shocks among many control groups to approximate the distribution of
the treated groups' shocks. The present approach instead uses the time-series
distribution of pre-treatment differential trends for the same treated and
control groups.

\medskip
A large literature uses pre-treatment trends to assess or relax parallel
trends. \citet{roth2022} studies inference conditional on passing a pre-trends
test. \citet{bilinskihatfield2018} and \citet{detteschumann2024} advocate
non-inferiority or equivalence tests, which ask whether pre-treatment
differences in trends are small enough to be substantively negligible.
\citet{rambachanroth2023} use restrictions on how violations of parallel
trends may evolve from pre- to post-treatment periods to construct robust
confidence sets. These papers use pre-trends primarily to diagnose, bound, or
extrapolate violations of parallel trends. In contrast, the proposed test
uses the empirical distribution of the differential pre-trends themselves as
the reference distribution for the final differential trend. While it allows for differential trends, it does so in a more restricted way than \citet{rambachanroth2023}.

\medskip
The procedure is also connected to end-of-sample stability and structural-break tests, particularly \citet{andrews2003}. From this perspective, the
null of no treatment effect says that the final differential trend is not a
distributional break relative to the pre-treatment differential-trend
process. The causal interpretation comes from the timing of treatment and the
additive decomposition $U_T=U_T(0)+\tau$.

\medskip
I am not aware of an earlier difference-in-differences paper that proposes
exactly the test considered here: ranking the absolute post-treatment
difference-in-differences within the time-series distribution of absolute
pre-treatment difference-in-differences while explicitly allowing a stable,
nonzero differential trend. The closest antecedent appears to be the general
counterfactual conformal framework of \citet{chernozhukov2021}, together with
the end-of-sample stability literature. This novelty claim should nonetheless
be treated as provisional rather than exhaustive.

\subsection{Related conformal-inference literature}

The procedure belongs to the broad family of conformal methods that calibrate
a new prediction error using historical errors. Classical conformal
prediction obtains finite-sample validity from exchangeability; see
\citet{vovkgammermanshafer2005} and \citet{leietal2018}. Time-series settings
require weaker approximate-exchangeability arguments or explicit adjustments
for dependence. \citet{chernozhukov2018} develop conformal methods for
dependent data using permutation schemes, while \citet{xu2021} construct
prediction intervals for dynamic time series using ensembles of leave-one-out
residuals.

\medskip
For causal inference, \citet{lei2021} construct conformal intervals for
counterfactual outcomes and individual treatment effects under potential-
outcome assumptions. Their goal differs from the present aggregate time-series
test, but the common principle is to turn uncertainty about an unobserved
counterfactual into a prediction problem and calibrate prediction errors
without imposing a fully parametric distribution.

\medskip
The closest reference remains \citet{chernozhukov2021}: their method explicitly
covers counterfactual and synthetic-control estimators, including
difference-in-differences, and permits weakly dependent time-series errors.
The contribution of the present specialization is to use differential trends
as the primitive conformity scores. Doing so makes it transparent that the
required stability restriction need not be parallel trends.

\end{document}